\documentclass[conference]{IEEEtran}

\usepackage{cite}
\usepackage{amsmath,amssymb,amsfonts}
\usepackage{amsthm}
\usepackage{graphicx}
\usepackage{textcomp}
\usepackage{enumitem}
\usepackage{booktabs}
\usepackage{threeparttablex}
\usepackage{algorithm}
\usepackage{algpseudocode}
\usepackage{tikz}
\usepackage{pgfplots}
\usepackage{soul,xcolor}
\usepackage{adjustbox}
\usepackage[hidelinks]{hyperref}

\pgfplotsset{compat=1.18}
\newtheorem{theorem}{Theorem}

\def\BibTeX{{\rm B\kern-.05em{\sc i\kern-.025em b}\kern-.08em
T\kern-.1667em\lower.7ex\hbox{E}\kern-.125emX}}

\begin{document}

\title{Plug-and-Play Quantum-Resistant BLE Pairing for Medical Implants via NFC Out-of-Band}

\author{
\IEEEauthorblockN{
Muhammad Asim Javaid\IEEEauthorrefmark{1},
Muhammad Shafay Tanveer Hassan\IEEEauthorrefmark{1},
Faaiz Umer\IEEEauthorrefmark{1},
Sotirios Panagiotou\IEEEauthorrefmark{4},\\
Christos Strydis\IEEEauthorrefmark{4},
Naveed Anwar Bhatti\IEEEauthorrefmark{1},
Muhammad Adeel Pasha\IEEEauthorrefmark{1},
Muhammad Ali Siddiqi\IEEEauthorrefmark{1}\IEEEauthorrefmark{4}
}
\IEEEauthorblockA{\IEEEauthorrefmark{1}Syed Babar Ali School of Science and Engineering, Lahore University of Management Sciences, Lahore, Pakistan}
\IEEEauthorblockA{\IEEEauthorrefmark{4}Department of Neuroscience, Erasmus University Medical Center, Rotterdam, The Netherlands}
}

\maketitle

\begin{abstract}
Bluetooth Low Energy (BLE) pairing establishes the cryptographic foundation for secure device communication. However, mainstream man-in-the-middle (MITM)-resistant pairing methods, such as Numeric Comparison and Passkey Entry, require user interfaces that implantable medical devices (IMDs) inherently lack, making Near-Field Communication (NFC)-assisted Out-of-Band (OOB) pairing an attractive alternative. Existing NFC-assisted OOB schemes authenticate a classical BLE key exchange that remains vulnerable to quantum attacks, while the NFC channel itself is susceptible to eavesdropping and active injection under stronger threat models. To address these limitations, we propose a lightweight, plug-and-play NFC-based OOB pairing protocol that performs a post-quantum key encapsulation mechanism (KEM) entirely over the NFC channel, ensuring that no shared secret is transmitted over NFC while reducing long-range radio-frequency (RF) exposure during pairing. The proposed protocol requires no modifications to the BLE stack and inherently resists RF battery-depletion attacks. Evaluation on an IMD-class proof-of-concept testbed demonstrates that post-quantum OOB pairing is practical on resource-constrained devices, reducing projected battery life by less than 0.57\% for Kyber-1024 and 0.56\% for FireSABER under an operational usage model relative to the MITM-vulnerable Just Works method.
\end{abstract}

\begin{IEEEkeywords}
IMD, Bluetooth Low Energy, post-quantum cryptography, KEM, out-of-band, NFC, battery DoS
\end{IEEEkeywords}

\section{Introduction}

Bluetooth Low Energy (BLE) employs user-mediated methods such as \emph{numeric comparison} or \emph{passkey entry} to prevent MITM attacks during pairing~\cite{BluetoothCoreV3H}. These approaches, however, assume that both devices provide a user interface (screen and/or keypad). Implantable medical devices (IMDs) are among the most security-sensitive BLE endpoints and typically lack any such interface, making such pairing methods infeasible (see Fig.~\ref{fig:overview}A). A seemingly straightforward alternative is to perform public-key cryptography (PKC) directly over the BLE radio channel for key establishment. However, this requires a public-key infrastructure (PKI) with certificates, which BLE pairing deliberately avoids since devices (including IMDs) typically lack Internet connectivity and cannot validate or update certificates, making such an approach infeasible in practice.

Consequently, the BLE specification defines an Out-of-Band (OOB) pairing method~\cite{BluetoothCoreV3H}, particularly suited for devices without a user interface. In OOB pairing, devices exchange authentication data over an auxiliary channel (e.g., NFC) and subsequently complete the BLE pairing procedure over the wireless link to establish an encrypted session (see Fig.~\ref{fig:overview}B). The security of this process therefore depends on both the OOB channel and the underlying key-establishment mechanism.
Among the available OOB technologies, Near-Field Communication (NFC) is the preferred choice due to its low energy consumption, short nominal range, and widespread support in smartphones and clinical programmers~\cite{BTSSP2020}. However, NFC communications remain susceptible to eavesdropping at greater-than-expected distances and, under stronger adversarial models, to active injection attacks~\cite{al2017towards,sun2020you}. Moreover, the BLE LE Secure Connections OOB model still relies on classical elliptic-curve cryptography for key establishment, leaving it vulnerable to future quantum adversaries.

\begin{figure}[!t]
    \centering
    \includegraphics[width=0.9\linewidth]{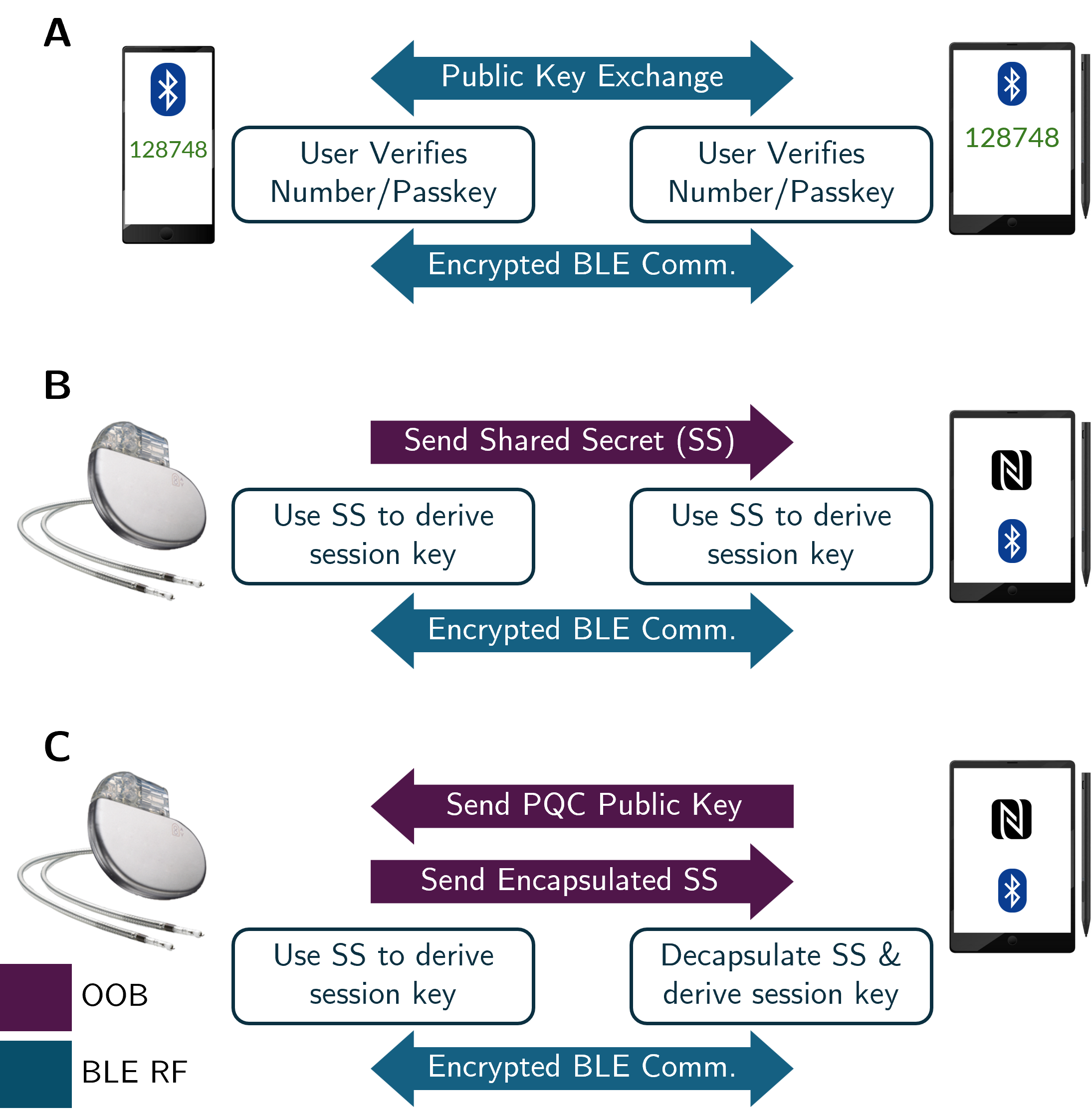}
    \caption{(A) BLE pairing with Numeric Comparison / Passkey Entry: User-mediated verification of public keys prevents MITM. (B) BLE pairing with IMD via NFC OOB: A shared secret is transferred over NFC, leaving it vulnerable to eavesdropping. (C) Proposed PQC-based NFC OOB pairing: The shared secret is established via KEM, ensuring secrecy and plug-and-play compatibility with BLE.}
    \label{fig:overview}
\end{figure}

To address these limitations, we propose a lightweight NFC-based OOB pairing protocol that performs a post-quantum key encapsulation mechanism (KEM) entirely over the NFC channel. This strengthens the OOB exchange against passive eavesdropping, renders active MITM infeasible under the stated NFC threat model via NFC's touch-to-access property, and reduces long-range RF exposure during pairing. The protocol is fully BLE stack-compatible, requiring no modifications to the BLE stack or Bluetooth Core Specification, as the OOB exchange is implemented entirely at the application layer.
We validate the proposed protocol on an IMD-class testbed using \textsf{CRYSTALS-Kyber}~\cite{avanzi2019crystals}, the basis of NIST's standardized ML-KEM~\cite{nistFIPS203}, and the Round 3 finalist \textsf{SABER}~\cite{d2018saber}.

In summary, this paper makes the following contributions:

\begin{enumerate}[label=\textbf{C\arabic*}.,
                  leftmargin=*,
                  itemsep=2pt,
                  topsep=2pt]
    \item We design a \textbf{lightweight BLE pairing protocol} for IMDs that uses a post-quantum KEM over the NFC OOB channel to prevent passive eavesdropping, strengthen resistance to active MITM attacks, and mitigate RF-based battery denial-of-service (DoS) attacks~\cite{siddiqi2021zero,siddiqi2019towards}.
    \item We implement the protocol on a real-world testbed demonstrating its \textbf{plug-and-play} compatibility with existing BLE stacks.
    \item We evaluate protocol \textbf{feasibility on resource-constrained IMDs} using lattice-based PQC schemes (\textsf{Kyber}, \textsf{SABER}), reporting execution time, energy consumption, and battery lifetime.
\end{enumerate}

\section{Background}
\label{sec:background}

\subsection{BLE Association Models}
\label{sec:ble_security}

BLE defines several \emph{association models}, i.e., mechanisms for establishing shared keying material during pairing, which is subsequently used to derive session keys for secure communication. The achieved security level depends on the devices' input/output (I/O) capabilities and the association model selected during feature exchange~\cite{BluetoothCoreV3H}. The main models are summarized as follows.

\emph{Just Works} establishes a shared key without user verification of the peer. It provides no protection against active MITM attacks because it lacks authentication. It is the default fallback for devices without displays or keypads, including resource-constrained endpoints, but is unsuitable for safety-critical IMDs due to its weak security.

In \emph{Passkey Entry}, one device displays a six-digit passkey that the user enters on the other device. This enables authenticated pairing and strong MITM resistance. However, it requires a display on one device and input capability on the other, making it impractical for IMDs without a user interface.

\emph{Numeric Comparison} displays a six-digit value on both devices, which the user confirms as matching. It provides strong MITM protection with less user effort than Passkey Entry, but similarly requires displays on both devices and is therefore unsuitable for IMDs.

Finally, the \emph{Out-of-Band (OOB)} model uses a secondary channel (e.g., NFC) to exchange authentication data used in key derivation. It eliminates the need for displays or keypads and is therefore attractive for IMDs. Its security, however, depends on the trustworthiness and attack resistance of the OOB channel.

\subsection{The Need for Post-Quantum Cryptography}
\label{sec:pqc_overview}

PQC studies cryptographic algorithms believed to remain secure against adversaries equipped with large-scale quantum computers. The urgency stems from the fact that widely deployed public-key schemes, RSA, Diffie-Hellman, and elliptic-curve cryptography (which is used in BLE), would be rendered insecure by quantum algorithms such as Shor's once scalable quantum hardware becomes available. This threat is compounded by ``store-now, decrypt-later'' adversaries who archive encrypted traffic today for future decryption~\cite{paar2024understanding}. However, PQC primitives generally incur higher computational and memory overheads, posing challenges for resource-constrained IoT and IMD platforms.

For bootstrapping on constrained devices, key encapsulation mechanisms (KEMs) are particularly suitable for establishing shared secrets from which symmetric keys are derived. Among these, \textit{lattice-based} KEMs offer comparatively efficient key establishment. In this work, we focus on Kyber and SABER\@. Kyber is the basis of NIST's standardized ML-KEM, while SABER was a Round 3 finalist. They reflect distinct lattice design approaches, Module Learning With Errors (Kyber) and Module Learning With Rounding (SABER), offering complementary perspectives on feasibility in IMD-class devices.

\subsection{Related Work}
\label{sec:related_work}

Considerable research effort has been devoted to securing IMDs over the past two decades, yet this has not been effectively translated into commercial practice~\cite{siddiqi2022securing}; one such shortcoming concerns the quantum readiness of the underlying communication standards. BLE LE Secure Connections pairing relies on Elliptic Curve Diffie-Hellman (ECDH). Although classically secure, ECDH is vulnerable to quantum adversaries, creating a long-term security gap for these devices. Recent work therefore explores lightweight PQC schemes that provide quantum resistance without overwhelming constrained platforms. Liu et al.~\cite{liu2024towards} propose a backward-compatible design that replaces ECDH with KYBER-512 in the BLE stack. Their approach maps key generation, encapsulation, and decapsulation onto GATT characteristics, while claiming interoperability with existing systems. However, the protocol still exchanges sensitive information such as the initial pairing request, acknowledgment, and public key over the long-range RF channel, leaving it vulnerable to active MITM attacks despite post-quantum security. Its use of a 404-byte maximum transmission unit (MTU) requires support for that negotiated size at both endpoints.

A natural progression is the use of OOB pairing, which shifts sensitive key exchange to an auxiliary channel, reducing exposure on the BLE RF link itself. Once established, the secure link supports normal BLE communication. Candidate OOB channels include wired links, body-coupled communication (BCC), optical signaling (e.g., LED/QR codes), human motion, and NFC, each offering different trade-offs in usability, hardware constraints, and threat assumptions~\cite{ban2018bcc,zhang2024oobkey,siddiqi2021securing,latvala2020evaluation,mao2017investigation,chung2025implantable,inyangson2025sok}. Among these, NFC is the most practical option for IMDs and constrained IoT devices. It enforces close proximity, is widely supported on smartphones and clinical programmers, and enables an intuitive touch-to-pair workflow. Consequently, it has become the de facto OOB channel in industry practice and joint NFC Forum/Bluetooth SIG specifications~\cite{BTSSP2020}.

In current (non-IMD) deployments, NFC is commonly used to convey BLE OOB authentication data without additional protection, relying on the short operating range of NFC to provide security. However, prior studies demonstrate passive eavesdropping beyond nominal distances~\cite{al2017towards} and active injection attacks at close proximity with specialized hardware~\cite{sun2020you}. Therefore, while NFC remains the most practical OOB channel for IMDs, its security assumptions require strengthening.

In summary, several BLE pairing methods have been explored in the literature~\cite{shi2025formal}; however, to the best of our knowledge, none demonstrates a quantum-resistant OOB pairing mechanism that leaves the BLE stack unmodified. Our approach is the first to achieve this, securing the pairing process itself against passive eavesdropping, rendering active MITM attacks infeasible under the stated NFC threat model, and mitigating RF-based battery DoS attacks.

\section{System and Threat Models}
\label{sec:model}

\subsection{System Overview}

We consider a system comprising two main components: an \textbf{IMD} and a \textbf{programmer} (e.g., tablet or smartphone), as illustrated in Fig.~\ref{fig:overview}C. The IMD is a resource-constrained endpoint, while the smartphone serves as a resource-rich interface for users and clinicians. Communication between these components occurs over two complementary channels: (i) NFC is used as a short-range OOB channel for bootstrapping cryptographic material during pairing, after which (ii) BLE provides the primary, higher-bandwidth link for ongoing encrypted data exchange.

\subsection{Threat Model}
\label{sec:threat_model}

We model the adversary as a probabilistic polynomial-time (PPT) algorithm $\mathcal{A}$ interacting with honest protocol participants over the BLE and NFC channels.

\subsubsection{Channel Model}

\paragraph{Bluetooth Channel.}
The Bluetooth channel is modeled under the standard Dolev--Yao adversarial setting. $\mathcal{A}$ has full control over this channel and may:
\begin{itemize}
    \item eavesdrop on all messages,
    \item block, delay, replay, or reorder transmissions,
    \item inject arbitrary protocol-level messages,
    \item initiate and interleave concurrent protocol sessions.
\end{itemize}

\paragraph{NFC Channel.}
Due to its physical proximity requirements, NFC communication requires close-range interaction. We assume that large-scale or remote physical-layer manipulation of the NFC carrier signal requires specialized hardware and close physical access beyond that of a typical remote adversary. Empirical studies report that passive NFC eavesdropping is feasible at distances of up to approximately 10\,m under favorable conditions~\cite{al2017towards}, while practical active injection and relay attacks are generally limited to distances below approximately 0.5\,m under laboratory conditions~\cite{sun2020you}, reflecting the substantially tighter proximity constraints of active attacks relative to passive eavesdropping. Accordingly, active physical-layer signal injection is considered infeasible under this model; since replaying a captured message or relaying a live exchange likewise requires actively transmitting into the NFC near-field at this same close range, we treat replay and relay as similarly infeasible against a remote adversary. We therefore assume that, within communication range, the adversary is limited to passive eavesdropping on NFC exchanges.

\subsubsection{Cryptographic Assumptions}
All employed cryptographic primitives (hash functions, MACs, authenticated encryption, and key derivation functions) satisfy their standard security definitions. Each NFC key exchange establishes fresh keying material. The lattice-based KEM is assumed to be IND-CCA secure against probabilistic polynomial-time adversaries.

\subsubsection{Security Arguments}

\begin{theorem}[Resistance to MITM on NFC KEM Exchange]
\label{theorem-mitm}
Under the above channel models and assumptions, a textbook MITM on the NFC-based post-quantum KEM is infeasible.
\end{theorem}
\begin{proof}[Proof Sketch]
A classical MITM attack on unauthenticated public-key exchange requires the adversary to intercept and substitute public keys, thereby establishing separate shared secrets with each honest party. Under our NFC threat model, the adversary may passively eavesdrop but cannot synthesize new protocol-compliant messages or perform real-time modification of NFC frames.

Since public-key substitution is excluded by assumption, the adversary cannot instantiate two independent KEM sessions with the honest devices. Moreover, by the IND-CCA security of the underlying lattice-based KEM, passive observation alone does not reveal the derived shared secret. Therefore, under the stated proximity-constrained NFC model, a textbook MITM attack on the NFC-based PQC KEM is infeasible.
\end{proof}

\begin{theorem}[Protection Against Remote RF Battery DoS]
\label{theorem-bdos}
Under the aforementioned channel models, a PPT adversary without close physical proximity cannot induce unbounded energy consumption on the IMD via the BLE interface.
\end{theorem}
\begin{proof}[Proof Sketch]
Protocol initiation is exclusively triggered via NFC, and the IMD enables BLE only after successful NFC interaction. The adversary has full control over BLE but NFC is proximity-constrained and excludes active signal injection, replay, and relay; hence a remote adversary cannot initiate NFC and therefore cannot activate BLE. Since BLE remains disabled without valid NFC initiation, repeated high-energy radio activity cannot be remotely triggered.

Within NFC range, an adversary is limited to passive eavesdropping and therefore cannot generate fresh valid session initiations on its own. Consequently, remote RF-based battery DoS is infeasible under the stated assumptions.
\end{proof}

\section{Proposed Protocol}
\label{sec:proposed_method}

\begin{figure}
    \centering
    \includegraphics[width=0.9\columnwidth]{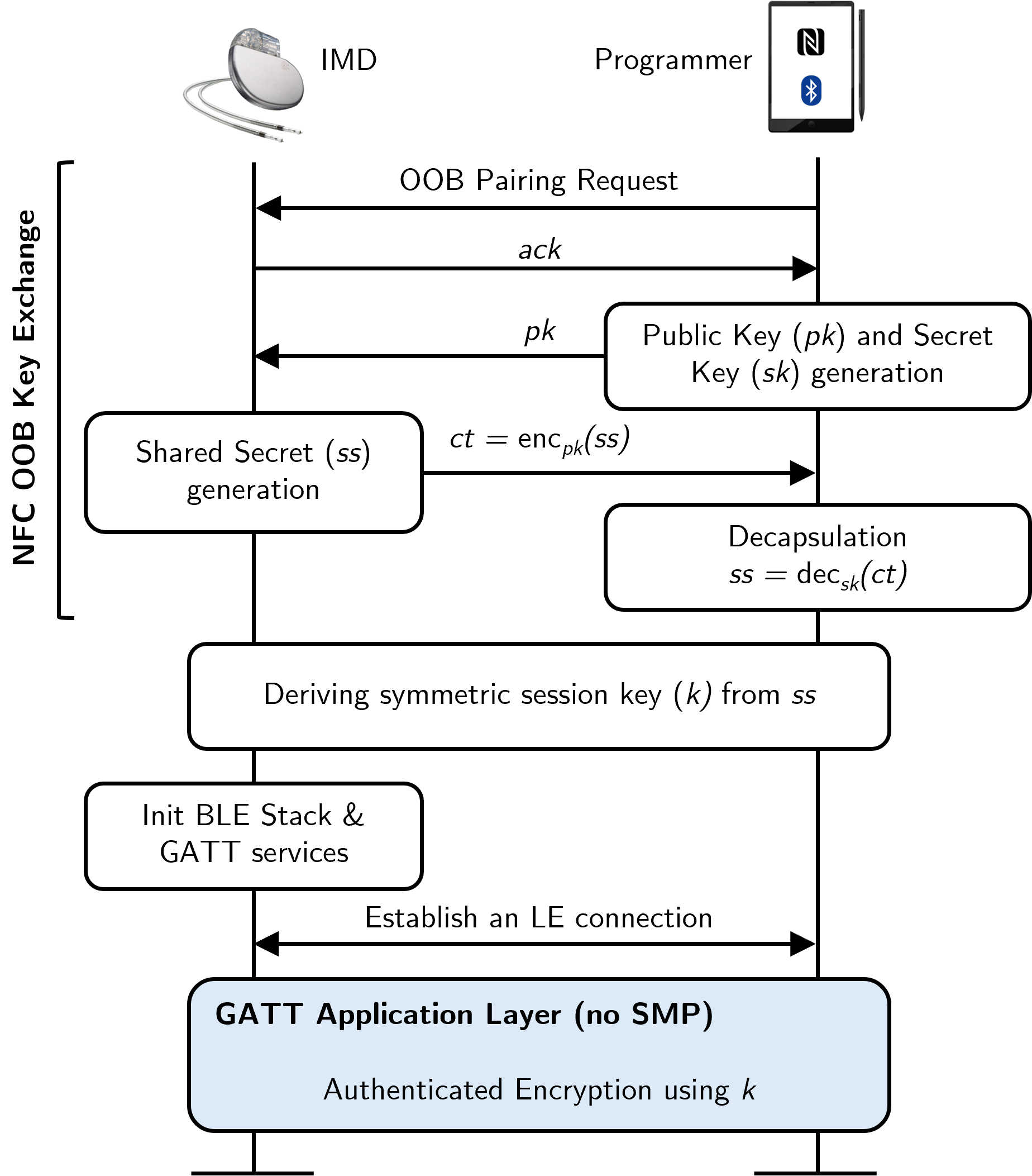}
    \caption{Proposed secure BLE communication protocol.}
    \label{fig:myplot}
\end{figure}

Our protocol first establishes a shared secret via a brief exchange over the NFC OOB channel, which is then used to generate a symmetric session key to secure the BLE communication (see Fig.~\ref{fig:overview}C). To preserve compatibility with existing BLE implementations, our design performs the entire post-quantum key encapsulation over the NFC OOB channel while leaving the BLE stack and Security Manager Protocol (SMP) unchanged. This approach avoids exposing public key exchange material over the long-range BLE radio channel and avoids dependence on larger negotiated BLE MTUs, both limitations of the KYBER-512-over-GATT design in~\cite{liu2024towards}. Consequently, the proposed protocol supports plug-and-play deployment with commodity BLE devices while providing quantum-resistant OOB bootstrapping. The protocol is illustrated in detail in Fig.~\ref{fig:myplot}.

\subsection{Public Key Transfer (Programmer \texorpdfstring{$\rightarrow$}{->} IMD)}
The protocol begins when the programmer sends an NFC application-selection request containing an application identifier (AID), which the IMD acknowledges. The programmer generates a post-quantum KEM key pair \((pk, sk)\) and sends the public key \(pk\) to the IMD for encapsulation. This design confines the IMD to a single encapsulation operation, leaving resource-intensive key generation and decapsulation to the programmer, while avoiding certificate overhead since active MITM attacks are infeasible (Theorem~\ref{theorem-mitm}).

Because PQC public keys exceed NFC payload limits, the key is fragmented into packets, each with a slim header for reassembly and duplicate/replay detection against a proximate attacker. A stop-and-wait protocol ensures the programmer sends one fragment at a time, waiting for IMD acknowledgment before proceeding (Algorithm~\ref{alg:nfc_stop_wait}). This guarantees complete reconstruction without exceeding the IMD's limited memory while preventing loss, duplication, or reordering.

\begin{algorithm}[t]
\small
\caption{Stop-and-Wait Fragment Transmission over NFC}
\label{alg:nfc_stop_wait}
\begin{algorithmic}[1]
\Require $M$, $C_{\max}$, timeout $\tau$, retry limit $R$, session $s$, type $\text{tag}$
\State Split $M$ into fragments $\{F_1, F_2, \ldots, F_n\}$, $|F_i| \le C_{\max}$
\For{$i = 1 \to n$}
    \State $H_i \gets (\text{tag}, |M|, \text{offset}_i, |F_i|, s)$
    \For{$r = 0 \to R$} \Comment{Initial attempt plus $R$ retries}
        \State Send $(H_i,F_i)$; await matching ACK/NACK or $\tau$
        \State On ACK: break; on NACK/timeout: continue
    \EndFor
    \State If no matching ACK: abort transfer
\EndFor
\Statex
\Comment{Receiver verification}
\State buffer $\gets \emptyset$; bytes\_received $\gets 0$
\For{each $(H_i, F_i)$ received}
    \If{header or payload invalid for this transfer}
        \State Send NACK; discard fragment
    \ElsIf{exact duplicate of an accepted fragment}
        \State Resend ACK; do not append
    \ElsIf{$\text{offset}_i =$ bytes\_received}
        \State Append $F_i$; bytes\_received $\gets$ bytes\_received$+|F_i|$; ACK
        \State If bytes\_received $=|M|$: deliver $M$ once
    \Else
        \State Send NACK; discard fragment
    \EndIf
\EndFor
\end{algorithmic}
\end{algorithm}

\subsection{Ciphertext Transfer (IMD \texorpdfstring{$\rightarrow$}{->} Programmer)}

After reassembling the public key, the IMD generates a shared secret $ss$ and encapsulates it using $pk$ to produce $ct = \text{enc}_{pk}(ss)$. If $ct$ exceeds the NFC frame size, it is fragmented into packets with headers specifying message type and total length for correct reassembly. The IMD sends fragments using the same stop-and-wait protocol (Algorithm~\ref{alg:nfc_stop_wait}), awaiting acknowledgment before sending the next. Once fully received, the programmer decapsulates $ct$ with $sk$ to recover $ss$.

\subsection{Key Derivation and Secure Communication}

The protocol establishes a shared secret through a post-quantum KEM exchange over the NFC out-of-band channel. Both endpoints derive a symmetric key $\mathit{k}$ from this secret using a key-derivation function and use authenticated encryption to protect application data transmitted over BLE. NFC carries the KEM public key and ciphertext, not $\mathit{k}$ itself. This OOB exchange bootstraps a secure channel at the Generic Attribute Profile (GATT) application layer and does not perform BLE SMP pairing. Android restricts OOB bond creation to system-privileged apps~\cite{androidBTPermissions}, while iOS exposes no pairing parameters through CoreBluetooth~\cite{appleCBDoc}. BLE remains disabled until the NFC key exchange completes.

\section{Implementation}
\label{sec:implementation}

\subsection{Hardware Setup}
The proof-of-concept IMD uses an ESP32-S3 connected to an NXP PN532 NFC transceiver via SPI. The ESP32-S3 provides BLE connectivity and sufficient processing for post-quantum operations under resource constraints, while the PN532 is a commodity NFC frontend typical in embedded systems. The programmer is a Samsung Galaxy A54 running Android~14, acting as NFC initiator and BLE central. Fig.~\ref{fig:experimental-setup} (left) shows the experimental setup.

\begin{figure}[!t]
    \centering
    \includegraphics[width=\columnwidth]{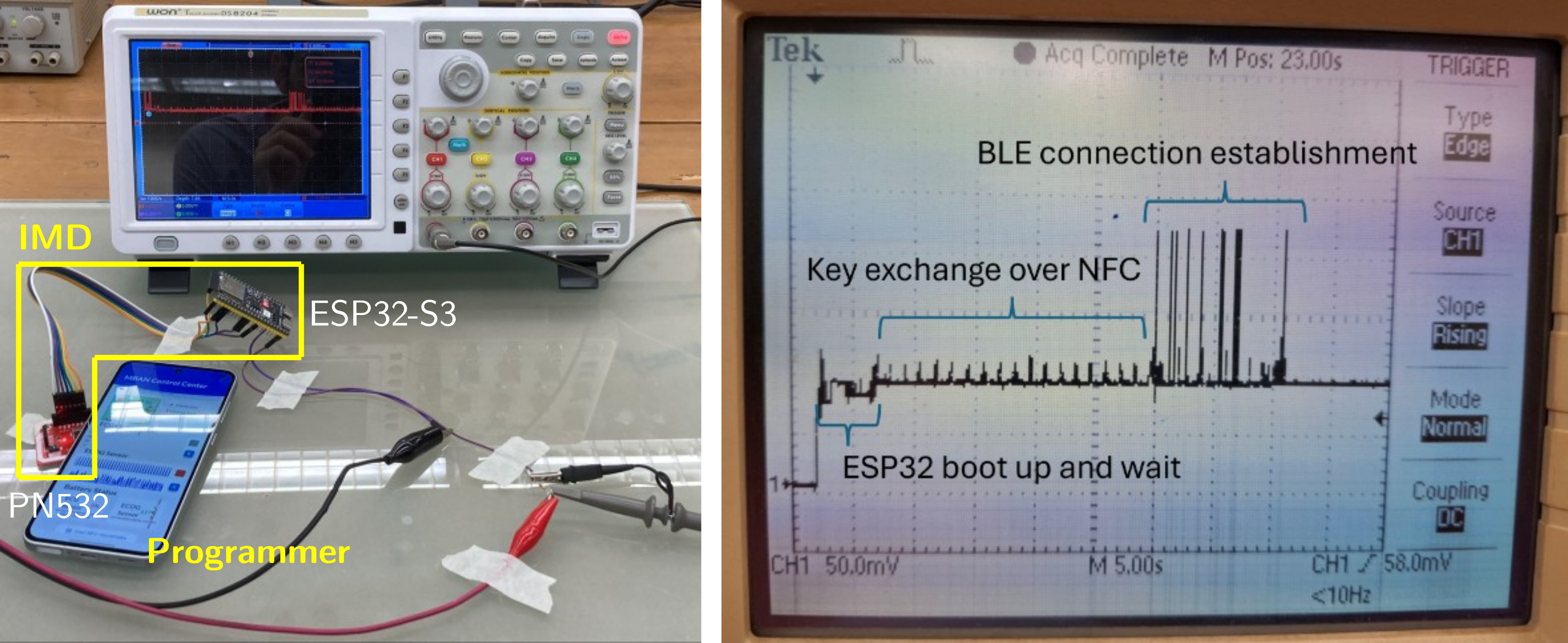}
    \caption{Left: Experimental setup for measuring IMD current during NFC-OOB key exchange and BLE data transmission. Right: Oscilloscope trace of the voltage drop across the series current-sense resistor.}
    \label{fig:experimental-setup}
\end{figure}

\subsection{Software Setup}
On the IMD side, the protocol is implemented using the built-in Bluedroid stack for BLE services. Cryptographic operations use the Open Quantum Safe \texttt{liboqs} library~\cite{liboqs} for Kyber and the reference SABER implementation from COSIC~\cite{saber}. On the programmer side, a mobile application was developed in Flutter~v3.19. Native Kotlin modules are invoked through platform channels to access Android's IsoDep API for NFC and the BluetoothGatt API for BLE services.

\subsection{Packet Fragmentation and Reassembly}

Each PQC public-key/ciphertext fragment (see Algorithm~\ref{alg:nfc_stop_wait}) includes a fixed 34-byte header: Message Tag (8 B), identifying the message type (public key or ciphertext); Total Message Length (4 B), specifying the total message length in bytes; Byte Offset (4 B), indicating the fragment's zero-based starting byte position within the message; Fragment Length (2 B), specifying the number of data bytes carried (0--200); and Session ID (16 B), associating fragments with the same key-exchange session. With up to 200 data bytes, each fragment occupies at most 234 bytes including its header, fitting within the 255-byte data field of a short command application protocol data unit (APDU).

\section{Results}\label{sec:results}

\subsection{Execution Time and Memory Footprint}

As reported in Table~\ref{tab:microbench} and illustrated in Fig.~\ref{fig:plots}A, end-to-end NFC--OOB session durations (Key ex.\ time) range from 11.4\,s to 20.7\,s for Kyber and 11.8\,s to 20.3\,s for SABER\@. The increase in duration correlates directly with the growth in ciphertext and public key sizes since larger messages require proportionally more NFC fragments to transmit. The encapsulation times for all variants are in the range of 10.0 to 47.2\,ms (cf. Table~\ref{tab:microbench}), representing less than 0.25\% of the corresponding session duration. These results indicate that session duration is dominated by NFC transfer and fragmentation overhead rather than by KEM computation. Moreover, all variants remain within the resource limits of IMD-class microcontrollers. The compiled firmware image sizes are in the range of 0.93 to 0.96\,MB, and reported DRAM usage remains below 45\%.

\begin{figure}
  \centering
  \includegraphics[width=0.85\columnwidth]{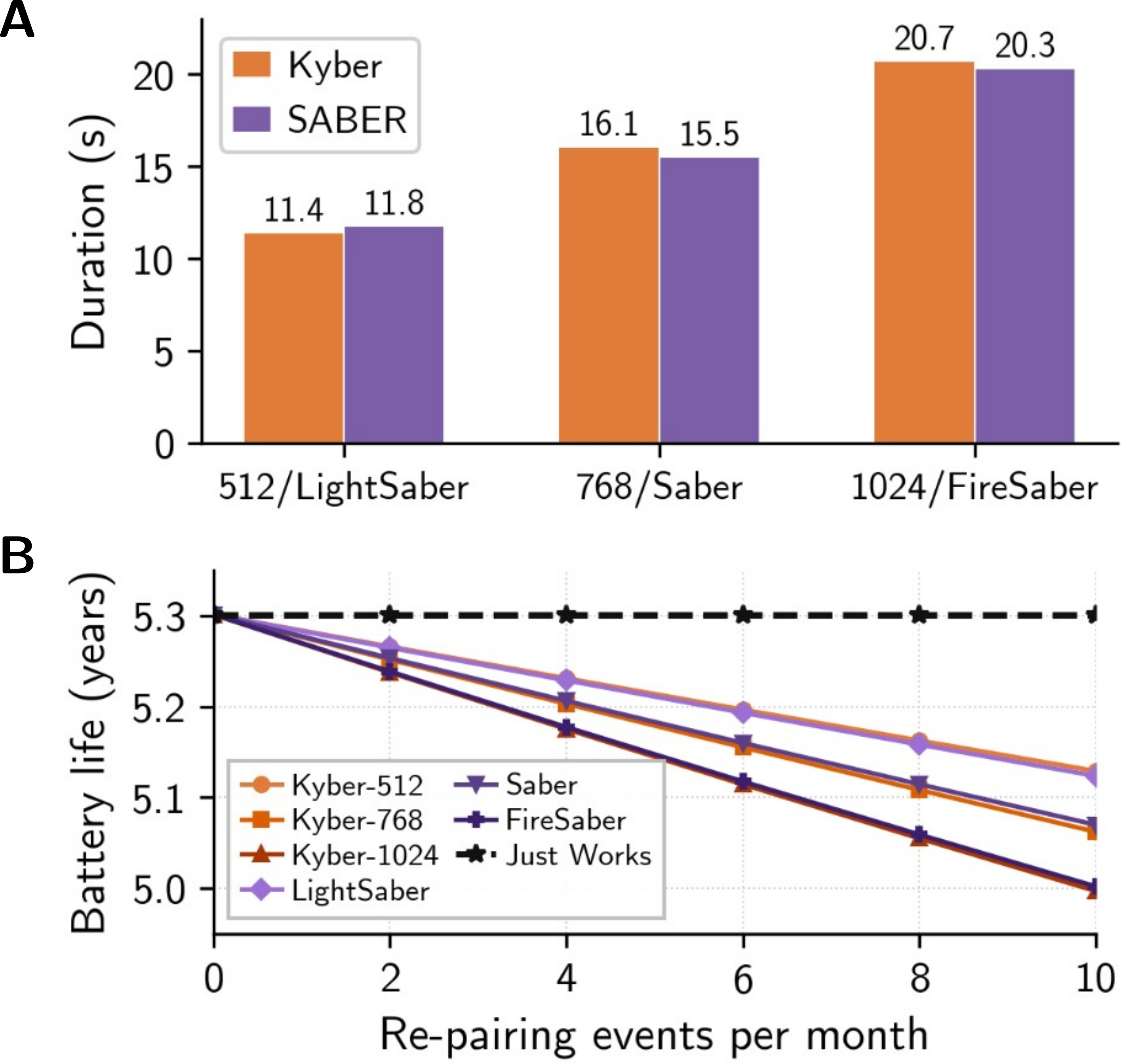}
  \caption{(A) OOB key exchange durations. (B) Battery life (1.5\,Ah) vs.\ frequency of re-pairing events.}
  \label{fig:plots}
\end{figure}

\begin{table*}[t]
\centering
\small
\caption{Micro-benchmark results for Kyber and SABER variants on ESP32-S3.}
\label{tab:microbench}
\begin{adjustbox}{max width=\textwidth}
\begin{tabular}{@{} l c c c c c c c @{}}
\toprule
Algorithm / Variant 
& $pk$ gen.\ \& trans.\ time (\textmu{}s) 
& $ct$ gen.\ time (\textmu{}s) 
& Encaps (\textmu{}s)
& Key ex.\ time (\textmu{}s)
& $ct$ size (B) 
& Image size (B) 
& DRAM usage (\%) \\
\midrule
Kyber-512    & 4\,652\,878 & 10\,508 & 10\,000 & 11\,421\,778 & 768  & 929\,756 & 42.66 \\
Kyber-768    & 6\,977\,926 & 15\,743 & 15\,300 & 16\,062\,045 & 1088 & 929\,676 & 43.73 \\
Kyber-1024   & 9\,304\,518 & 22\,517 & 22\,000 & 20\,712\,095 & 1568 & 929\,772 & 44.94 \\
\midrule
LightSABER  & 4\,789\,729 & 17\,447 & 16\,800 & 11\,792\,524 & 672  & 957\,548 & 42.81 \\
SABER        & 6\,085\,585 & 30\,733 & 30\,100 & 15\,522\,213 & 992  & 957\,236 & 43.84 \\
FireSABER   & 8\,480\,967 & 47\,820 & 47\,200 & 20\,342\,097 & 1568 & 957\,360 & 44.88 \\
\bottomrule
\end{tabular}
\end{adjustbox}
\end{table*}

\subsection{Energy Consumption and Autonomy}
\label{sec:energy}
The oscilloscope trace in Fig.~\ref{fig:experimental-setup} (right) records the voltage drop $V_{\mathrm{sense}}(t)$ across a 1\,$\Omega$ precision current-sense resistor placed in series with the IMD supply. We converted each sample to current using Ohm's law, $I(t)=V_{\mathrm{sense}}(t)/1\,\Omega$, and time-averaged the samples belonging to each labeled protocol phase. With the ESP32-S3 clocked at 80\,MHz, this procedure yielded average currents of 25\,mA during NFC operation and 35\,mA during BLE operation. Fig.~\ref{fig:plots}B shows the resulting projected battery lifetimes for a 1.5\,Ah, 3.3\,V battery, assuming a typical pacemaker stimulation energy of 20~\textmu{}J per heartbeat, an average heart rate of 70 beats/min~\cite{siddiqi2021securing}, and one 60\,s BLE session per day. This operational model assumes that the established symmetric key is retained across daily BLE sessions and refreshed through a monthly NFC exchange. Compared with the MITM-vulnerable Just Works baseline, the quantum-resistant NFC-OOB protocol reduces projected battery life by less than 0.57\% for Kyber-1024 and 0.56\% for FireSABER under this one-re-pairing-event-per-month operational workload.

\section{Discussion}
\label{sec:discussion}

The prototype demonstrates NFC-assisted post-quantum key establishment without BLE stack modifications. Its 11.4--20.7\,s exchange duration is suited to infrequent, supervised medical procedures. KEM encapsulation accounts for less than 0.25\% of this duration, making NFC transport the main optimization target. Larger NFC transfer payloads, higher supported bit rates, and fewer acknowledgment round trips could reduce latency; hardware acceleration would provide a smaller, complementary improvement.

Under the stated NFC threat model (Theorem~\ref{theorem-mitm}), close physical access to the patient is treated as a basis for trust, and NFC touch-to-access provides implicit authorization to initiate pairing without pre-provisioned programmer credentials or certificate exchange. Keeping BLE disabled until NFC initiation also limits exposure to remote RF-based battery DoS (Theorem~\ref{theorem-bdos}).

\section{Conclusion}
\label{sec:conclusion}

We have presented a lightweight, plug-and-play post-quantum pairing protocol for BLE-enabled IMDs that leverages NFC as an out-of-band channel. By integrating lattice-based KEMs (Kyber and SABER) into the NFC-assisted key exchange without modifying the BLE stack, the approach protects against passive eavesdropping, renders active MITM attempts infeasible under the stated NFC threat model, and mitigates RF-based battery DoS attacks, while imposing minimal overhead on resource-constrained IMDs. Experimental results on an IMD-class testbed demonstrate feasible execution times, memory footprints, and energy consumption, with a projected battery-life reduction below 0.57\% for Kyber-1024 and 0.56\% for FireSABER under operational usage assumptions. As session duration is dominated by NFC transfer rather than KEM computation, future work will focus on reducing transport latency through optimized NFC framing.

\bibliographystyle{IEEEtran}
\bibliography{references}

\end{document}